\documentclass[12pt,onecolumn, doublespace, draftcls]{article}

\usepackage{fullpage}
\usepackage{amsmath, amssymb, enumerate, color, graphicx, bm, url, mathbbol, amsthm, caption, subcaption}
\newtheorem{theorem}{Theorem}[section]
\newtheorem{proposition}[theorem]{Proposition}

\newtheorem{lemma}[theorem]{Lemma}

\newtheorem{remark}[theorem]{Remark}

\numberwithin{equation}{section}

\DeclareMathOperator*{\E}{\mathbb{E}}

\DeclareMathOperator*{\Id}{Id}

\DeclareMathOperator*{\rank}{rank}

\def \R {\mathbb{R}}

\def \P {\mathbb{P}}

\def \NN {\mathcal{N}}

\def \F  {{F}}

\def \eps {\varepsilon}

\def \z {\zeta}

\def \moo| {\langle}
\def \< {\langle }
\def \> {\rangle }
\def \^ {\widehat}

\newcommand{\yp}[1]{\textcolor{black}{#1}}

\newcommand{\rg}[1]{\textcolor{black}{#1}}
\newcommand{\norm}[1]{\left \|#1\right \|}
\newcommand{\zeronorm}[1]{\norm{#1}_0}
\newcommand{\onenorm}[1]{\norm{#1}_1}
\newcommand{\twonorm}[1]{\norm{#1}_2}

\newcommand{\opnorm}[1]{\norm{#1}}

\newcommand{\abs}[1]{\left | #1 \right |}

\renewcommand{\Pr}[1]{\P \left\{ #1 \rule{0mm}{3mm}\right\}}

\newcommand{\vect}[1]{\bm{#1}}
\newcommand{\mat}[1]{\bm{#1}}

\def \x {\vect{x}}
\def \y {\vect{y}}
\def \X {\mathcal{X}}
\def \A {\mat{A}}
\def \H {\mat{H}}
\def \V {\mat{V}}
\def \z {\vect{z}}

\def \OM {\mat{\Omega}}

\def\cosp {\ell} 

\usepackage{epstopdf}

\title{On the Effective Measure of Dimension \\ in the Analysis Cosparse Model}

\author{Raja~Giryes\thanks{R. Giryes is with the Department
of Electrical and Computer Engineering, Duke University, Durham,
NC, 27708 USA (e-mail: raja.giryes@duke.edu).},
        Yaniv~Plan\thanks{Y. Plan is with the Department of Mathematics, University of British Columbia, 1984 Mathematics Road, Vancouver, BC V6T 1Z2, Canada (e-mail: yaniv@math.ubc.ca).},
        and~Roman~Vershynin\thanks{R.~Vershynin is with the Department of Mathematics, University of Michigan, 530 Church
St., Ann Arbor, MI 48109, U.S.A (e-mail:romanv@umich.edu).}
\thanks{Manuscript received October 3, 2014. RG is partially supported by AFOSR. 
YP is partially supported by an NSF Postdoctoral Research Fellowship under award No.
1103909. RV  is partially supported by NSF under grant DMS 1265782,  by USAF
under grant FA9550-14-1-0009, and by Humboldt Research Foundation.}
}

\begin{document}
\maketitle

\begin{abstract}
Many applications have benefited remarkably from low-dimensional models in the recent decade.
The fact that many signals, though high dimensional, are intrinsically low dimensional has given the possibility to recover them stably from a relatively small number of their measurements. For example, in \textit{compressed sensing} with the standard (synthesis) sparsity prior and in \textit{matrix completion}, the number of measurements needed is proportional (up to a logarithmic factor) to the signal's manifold dimension. 

Recently, a new natural low-dimensional signal model has been proposed: the cosparse analysis prior.    
In the noiseless case, it is possible to recover signals from this model, using a combinatorial search, from a number of measurements proportional to the signal's manifold dimension. 
However, if we ask for stability to noise or an efficient (polynomial complexity) solver, all the existing results demand a number of measurements which is far removed from the manifold dimension, sometimes far greater. Thus, it is natural to ask whether this gap is a deficiency of the theory and the solvers, or if there exists a real barrier in recovering the cosparse signals by relying only on their manifold dimension.  
Is there an algorithm which, in the presence of noise, can accurately recover a cosparse signal from a number of measurements proportional to the manifold dimension?  In this work, we prove that there is no such algorithm.  Further, we show through numerical simulations that even in the noiseless case convex relaxations fail when the number of measurements is comparable to the manifold dimension. This gives a practical counter-example to the growing literature on compressed acquisition of signals based on manifold dimension.
\end{abstract}

\section{Introduction}

Low-dimensional signal models have played an important role in many  signal processing applications in the recent decade, where in many cases the use of these models has provided state of the art results \cite{Bruckstein09From, Romberg08Imaging, Willett11Compressed, Elad10Sparse}. All have relied on the fact that the treated signals, which have high ambient dimension, reside in a low dimensional manifold \cite{Lu08Theory, Blumensath09Sampling, Eldar12Compressed, Foucart13Mathematical, Plan14High, Vershynin14High}, or a union of manifolds, e.g., a union of subspaces.

The fact that a signal belongs to a low dimensional manifold may make it possible to recover it from few measurements. This is exactly the essence of the compressed sensing \cite{Donoho06Compressed, Candes06Near} and matrix completion \cite{Candes09Exact} problems.

Our core problem is to recover an unknown signal $\x \in \R^d$ from a  given set of its linear measurements 
\begin{eqnarray}
\y = \A \x + \z,
\end{eqnarray}
where $\A \in \R^{m \times d}$ is the measurement matrix with $m \ll d$ and $\z$ is an additive noise which can be either adversarial with bounded energy \cite{Candes06Near, Donoho06Stable, Candes05Decoding} or random with a certain given distribution, e.g.,  Gaussian \cite{Candes07Dantzig}. 

In the standard compressed sensing problem the low dimensionality of $\x$ is modeled using the synthesis sparsity model. The signal $\x$  is assumed to be either explicitly sparse, i.e., has a small number of non-zeros, or with a sparse representation under a given dictionary $\mat{D} \in \R^{d \times n}$. 

If every group of $2k$ columns in $\A\mat{D}$ (if $\x$ is explicitly sparse then $\mat{D} = \Id$) are independent, where $k$ is the sparsity of the signal, then $\x$ can be uniquely recovered from $\y$ in the noiseless case ($\z = 0$) \cite{Donoho03Optimal, Giryes13CanP0} using a combinatorial search. 

The above observation shows that it is possible to recover the signal with a number of measurements of the order of its manifold dimension $k$. However, the \rg{combinatorial} search is not feasible for any practical size of $d$ \cite{NP-Hard}. Therefore several relaxation techniques have been proposed \cite{Candes06Near, Chen98overcomplete, MallatZhang93, Needell10Signal, Needell09CoSaMP, Dai09Subspace,Blumensath09Iterative,Foucart11Hard}.
It has been shown that if $\A\mat{D}$ is a subgaussian random matrix
or a partial Fourier matrix then it is possible using these practical methods to recover $\x$ from $\y$ using only $O(k \log^c(n))$ measurements ($c$ is a given constant) \cite{Candes05Decoding, Rudelson06Sparse}. 
Note that up to the log factor, the number of needed measurements is of the order of the manifold dimension of the signal.

Noise is easily incorporated into these results.  The same number of measurements guarantees robustness to adversarial noise. 
Further, in the case of random white Gaussian noise the recovery error turns to be of the order of $k \sigma^2 \log(n)$ \cite{Candes07Dantzig,Bickel09Simultaneous, Giryes12RIP}, where $k\sigma^2$ is roughly the energy of the noise in the low dimensional subspace the signal resides in.  

Similar results hold outside of the standard (synthesis) sparsity model of compressed sensing.  Indeed, in the \textit{matrix completion} problem, in which the signal is a low-rank matrix, once again one may reconstruct the matrix from a number of measurements proportional to the manifold dimension \cite{Candes09Exact, Candes10Matrix}.  Further, there exists a body of literature giving an abstract theory of signal reconstruction based on manifold dimension \cite{Vershynin14High, Eldar12Uniqueness, baraniuk2009random, wakin2010manifold, yap2011stable, eftekhari2013new}.  
 
\subsection{The Cosparse Analysis Model}
Recently, a new signal sparsity framework has been proposed: the cosparse analysis model \cite{elad07Analysis, Nam12Cosparse} that looks at the behavior of the signal after applying a given operator $\OM \in \R^{p \times d}$ on it.  We introduce it with an important example: the vertical and horizontal finite difference operator (2D-DIF).  The motivation for the usage of this operator is that in many cases the image is not sparse but its gradient is, so the signal ``becomes sparse'' only after the application of 2D-DIF. Therefore it is common to represent the structure of an image through its behavior after the application of the 2D-DIF operator.  
Note that 2D-DIF, when used with $\ell_1$-minimization (See \eqref{eq:analysis_l1_noiseless} hereafter), corresponds to the anisotropic two dimensional total variation (2D-TV) \cite{Needell13Stable}.

For simplicity, we suppose that $d$ is a square number and consider an $\sqrt{d} \times \sqrt{d}$ signal matrix $X$, e.g., an image.  Then the horizontal differences operator, $\H$, and the vertical finite differences operator, $\V$, are defined as follows:
\begin{eqnarray}
& \H(X)_{i,j} := X_{i,j} - X_{i, j+1} \\    & \V(X)_{i,j} := X_{i,j} - X_{i+1, j}
\end{eqnarray}
with addition of indices being done mod $\sqrt{d}$.  One may then unfold the signal $X$ into vector form, and correspondingly represent $\H$ and $\V$ as $d \times d$ matrices.  Thus, we take $\OM \in \R^{2d \times d}$, in the 2D-DIF model, to be the horizontal differences matrix stacked on top of the vertical differences matrix .  

In the general case, for a given  operator $\OM \in \R^{p \times d}$, the cosparse analysis model  assumes that $\OM \x$ should be sparse.
The subspace in which the signal resides is characterized by the zeros in $\OM\x$. The number of zeros in $\OM\x$ is denoted as the cosparsity of the signal.  
Each zero entry characterizes a row in $\OM$ to which the signal is orthogonal. Denoting by $T$ the unknown support (the locations of the non-zero entries) of $\OM\x$  and by $\OM_T$ the submatrix of $\OM$ restricted to the rows in $T$ we have that the subspace of $\x$ is the one orthogonal to the subspace spanned by the rows of $\OM_{T^c}$.  In what follows, define 
\begin{eqnarray}
K_T := \{ \x \in \R^d : \OM_{T^c} \x = 0\}
\end{eqnarray}
to be such a subspace. \rg{Notice that the dimension of $K_T$, which we denote by $b$, equals  $d- \rank(\OM_{T^c})$.}
 In general, $T$ is not known, and so it is natural to assume a signal structure of the following form
\begin{equation}
\label{eq:signal structure}
K_b := \bigcup_{\text{dim}(K_T) = b} K_T.
\end{equation}
Note that this is a finite union of $b$-dimensional subspaces.  \rg{As each signal in $K_b$ belongs to a subspace of dimension $b$, we say that $K_b$ has a {\em manifold dimension} $b$.}\footnote{\yp{See \cite{lee2003smooth} for a definition of manifolds and manifold dimension.  We note that $K_b$ is technically not a manifold, but this can be easily remedied.  Consider the slightly smaller set: $K'_b = K_b \backslash \bigcup K_T \cap K_{T'}$ in which we take the union over all $T \neq T'$ such that $\text{dim}(K_T) = \text{dim}(K_{T'}) = b$.  $K'_b$ is a manifold, and none of the proofs would change under this alternative definition of the signal set.  Nevertheless, we define $K_b$ as in Equation \eqref{eq:signal structure} for simplicity of presentation.}}

Recent literature \cite{Needell13Stable, Candes11Compressed, Liu12Compressed,  Giryes13Greedy, Kabanava13Analysis, Giryes13TDIHT} shows that $\x$ may be reconstructed efficiently and stably from $O(\abs{T} \log(p) )$ random linear measurements.  How does this compare to the manifold dimension of the signal?  
Assume $\abs{T}$ is fixed and $\OM$ is in general position, e.g., each entry is taken from a Gaussian ensemble. Then $\x \in K_b$, where $b = (d - p + \abs{T})_+$ (as $\x$ is orthogonal to $p - \abs{T}$ rows in $\OM$). In particular, if \rg{$\abs{T} = p-d+1$} then $b=1$ and therefore $\x$ resides in a $1$-dimensional subspace. Surprisingly, modern theory requires more than \rg{$\abs{T} = p-d+1$} measurements for recovering the signal, i.e., more measurements than the ambient dimension \rg{in the case $ p \ge 2d$}.  

\rg{This behavior is not unique only to the case of $\OM$ in general position. For example, consider the 2D-DIF model with $b=2$. In this case $K_2$ consists of images with only two connected components.\footnote{All images in the subspace $K_T$ have the same pattern of edges, defined by the indices of $T$ (two adjacent pixels have an edge when they may take different values).  The edges of $T$ separate the image into connected components.   
If there are only two connected components it is enough to use only two numbers, which set the grey value of each component, to define each image in this subspace.  In this case, $K_T$ is a two-dimensional subspace.} However, also in this case the current theory requires the number of measurements to be proportional to $\abs{T}$, the number of edges in these images, which might be much larger than $2$. Notice that for the same manifold dimension, we may have different number of edges in different images. 
For example, in Fig.~\ref{fig:blob} the number of edges is roughly proportional to $\sqrt{d}$ and in Fig.~\ref{fig:packingpic_oneconfig} it is roughly proportional to $d$.
}


\rg{Following the above two examples,} it is natural to ask whether the state of the art theory may be improved. 
Indeed, let $\A \in \R^{m \times d}$ be a Gaussian matrix and observe that by solving 
\begin{eqnarray}
\label{eq:analysis_l0_noiseless}
\min_{\x'}\zeronorm{\OM\x'} & s.t.& \y = \A\x',
\end{eqnarray}
where $\norm{\cdot}_0$ is the $\ell_0$ pseudo-norm that counts the number of non-zeros in a vector, one may recover any $\x \in K_1$ using only two measurements. In the general case there is a need for $2 b$ measurements to recover a signal in $K_b$ \cite{Nam12Cosparse}.
However, solving~\eqref{eq:analysis_l0_noiseless} is NP-hard and requires a \rg{combinatorial} search. Thus, there is a large gap between the theory for tractable, stable signal recovery and what can be done by combinatorial search with noiseless measurements.



\subsection{Our Contribution}

With these observations before us, it is natural to ask whether the gap between the required number of measurements and the manifold dimension is a deficiency of the utilized approximation strategies and the used proof techniques, or \rg{whether} there exists a real barrier with recovering the cosparse signals by relying only on the manifold dimension.  \yp{Is there any algorithm that can robustly reconstruct a cosparse signal from a number of measurements proportional to the manifold dimension $b$?}

This paper addresses this question by considering the effect of Gaussian noise.  We show that unless the number of measurements is much larger than the manifold dimension, there is no estimator that can stably reconstruct the signal.  We show this for two different analysis dictionaries: 1) the vertical and horizontal finite difference operator and 2) a random Gaussian matrix.  We show that when $m < d < p$, the error must be exponentially larger than the noise no matter what estimator is used.  


  We state our two main theorems below.
\begin{theorem}
\label{thm:TV_main_theorem}
Let $\OM$ be the 2D-DIF operator and let $K_2$ be the union of $2$-dimensional subspaces generated by this matrix.  
Suppose $\y = \A \x + \z$ for some $\x \in K_2$, $\A \in \R^{m \times d}$ with $\opnorm{\A} \leq 1$, and $\z \sim \NN(0, \sigma^2\cdot \Id)$.  Then for any estimator $\hat{\x}(\y)$ we have
\[\max_{\x \in K_2} \E \twonorm{\hat{\x} - \x} \geq C \sigma \exp(c d/m).\]
\end{theorem}

\begin{theorem}
\label{thm:Gaussian_main_theorem}
Let $\OM \in \R^{p \times d}$ be a matrix with independent standard normal entries and let $K_1$ be the union \rg{of} $1$-dimensional subspaces generated by this matrix.  Suppose $\y = \A \x + \z$ for some $\x \in K_1$, $\A \in \R^{m \times d}$ with $\opnorm{\A} \leq 1$, and $\z \sim \NN(0, \sigma^2\cdot \Id)$.  Then for any estimator $\hat{\x}(\y)$ we have
\[ \max_{\x \in K_1} \E \twonorm{\hat{\x} - \x} \geq 
 C \sigma \exp\left(C\frac{d-1}{m}\left(1 - \frac{d-2}{p}\right)\right).\]
\end{theorem}

The theorems are proven in Section~\ref{sec:pack_construct} by providing packings for the sets $K_1$ and $K_2$, and combining these with a hypothesis testing argument.  

We can say that both theorems show that $O(d)$ measurements are needed for any algorithm to get signal reconstruction without incurring a huge error if $p \ge 2d$; the latter assumption is implicit in Theorem~\ref{thm:TV_main_theorem} where $p=2d$. \rg{Notice that in Theorem~\ref{thm:Gaussian_main_theorem}, if $d < p < 2d$ then we need $m = O(p-d)$, which is still remarkably larger than the manifold dimension $b=1$.}
Remarkably, both theorems lower bound the efficacy of any estimator, tractable or not, and thus show that even the performance of $\ell_0$ minimization is not characterized well by the manifold dimension of the signal.
Therefore we may conclude that the fact that the needed number of measurements is of the order of $\abs{T}$  is not a result of a flaw in the existing reconstruction guarantees or a problem with the studied methods.  
We perform several experiments to demonstrate this fact and show that indeed, the size of $\abs{T}$ (determined by the cosparsity of the signal, the number of zeros in $\OM\x$) is a better measure for its compressibility than its manifold dimension.

%

\subsection{Organization}

The paper is organized as follows. 
 Section~\ref{sec:pack_construct} gives the proofs of Theorems \ref{thm:TV_main_theorem} and  \ref{thm:Gaussian_main_theorem}.
In Section~\ref{sec:exp} we demonstrate the lower bounds we have developed through several experiments that use the $\ell_1$-minimization technique with a Gaussian matrix and the 2D-DIF operator as the analysis dictionary. In Section~\ref{sec:conc} we discuss the implications of the derived results and conclude the work.

\section{Proofs of Theorems \ref{thm:TV_main_theorem}  and \ref{thm:Gaussian_main_theorem}}
\label{sec:pack_construct}

Both of our main theorems are proven by construction of random packings, followed by a hypothesis testing argument.  Recall that a packing $\X \subset K$ with $\ell_2$ balls of radius $\delta$ is a set satisfying $\twonorm{x - y} \geq \delta$ for any $\x, \y \in \X$ with $\x \neq \y$.  We denote by $P(K, \delta)$ the maximal cardinality of such a set, i.e., the \textit{packing number}.

We now gather supporting lemmas and then put them together in Section \ref{sec:finals steps}.
We will construct a packing when $\OM$ is the vertical and horizontal differences operator and when $\OM$ is Gaussian.  In both cases, we will make a random construction using the following observation.

\begin{lemma}[Random packing]
\label{lem:random packing}
Let $\F$ be a distribution supported on some set $K \in \R^n$.  Let $\x, \x'$ be independently chosen from $\F$. Suppose that
\[\Pr{\twonorm{\x - \x'} < \delta} \leq \eta\]
for some $\eta, \delta > 0$.  Then,
\[P(K, \delta) \geq \eta^{-1/2}.\]
\end{lemma}
\begin{proof}
Let $q$ be $\eta^{-1/2}$ rounded up to the next nearest integer.  Pick $q$ points independently from $K$.  Then, by considering each of the ${q \choose 2}$ pairs and using the union bound we have
\begin{eqnarray}
\Pr{\min_{\x \neq \x'} \twonorm{\x - \x'} < \delta} \leq {q \choose 2} \eta 
 < \frac{(\eta^{-1/2} + 1) \eta^{-1/2}}{2} \cdot \eta < 1
\end{eqnarray}
so long as $\eta < 1$.  The minimum above is taken over all $\x \neq \x'$ in the $q$ random points.  Thus, with probability greater than 0, the $q$ points satisfy
\[\twonorm{\x - \x'} \geq \delta  \qquad \text{for } \x \neq \x'.\]
Thus, there must exist at least one such arrangement of points, making the requisite packing.
\end{proof}

\subsection{Packing when $\OM$ is the vertical and horizontal differences operator}

%
%
%
%
We construct a packing for the set $K := K_2 \cap S^{d-1}$ when $\OM$ is the vertical and horizontal finite differences operator.

\begin{lemma} Suppose that $d \geq 64$ is a square number.  Then
\label{lem:pack finite differences}
\[P(K, 1/2) \geq \exp(d/64).\]
\end{lemma}

\begin{proof}
$\OM$ is most simply visualized acting on images.  Thus, set $n:=\sqrt{d}$ to be the number of vertical or horizontal pixels in an image.  We make the corresponding abuse of notation and take $\OM: \R^{n \times n} \rightarrow \R^p$ and $K \subset \R^{n \times n}$.  Each zero entry of $\OM \x$ forces two adjacent entries of $\x$ to be equal.  Thus, the set $K$ is precisely the set of \rg{normalized} images $\x$ composed of two connected components, $\x_1$ and $\x_2$, with $\x$ constant on each component (See Figure \ref{fig:blob}).  Our question reduces to constructing a packing for pairs of blobs in $\R^{n \times n}$.  We restrict our attention to blobs with the pictorial representation of Figure \ref{fig:packingpic}.
\rg{Naturally, our packing is "non-exhaustive", as it aims at enumerating a large number of images (and not all images) that reside in $K$.}

\begin{figure}
	\centering
		\begin{subfigure}[b]{.25 \textwidth}
			\includegraphics[width = \textwidth]{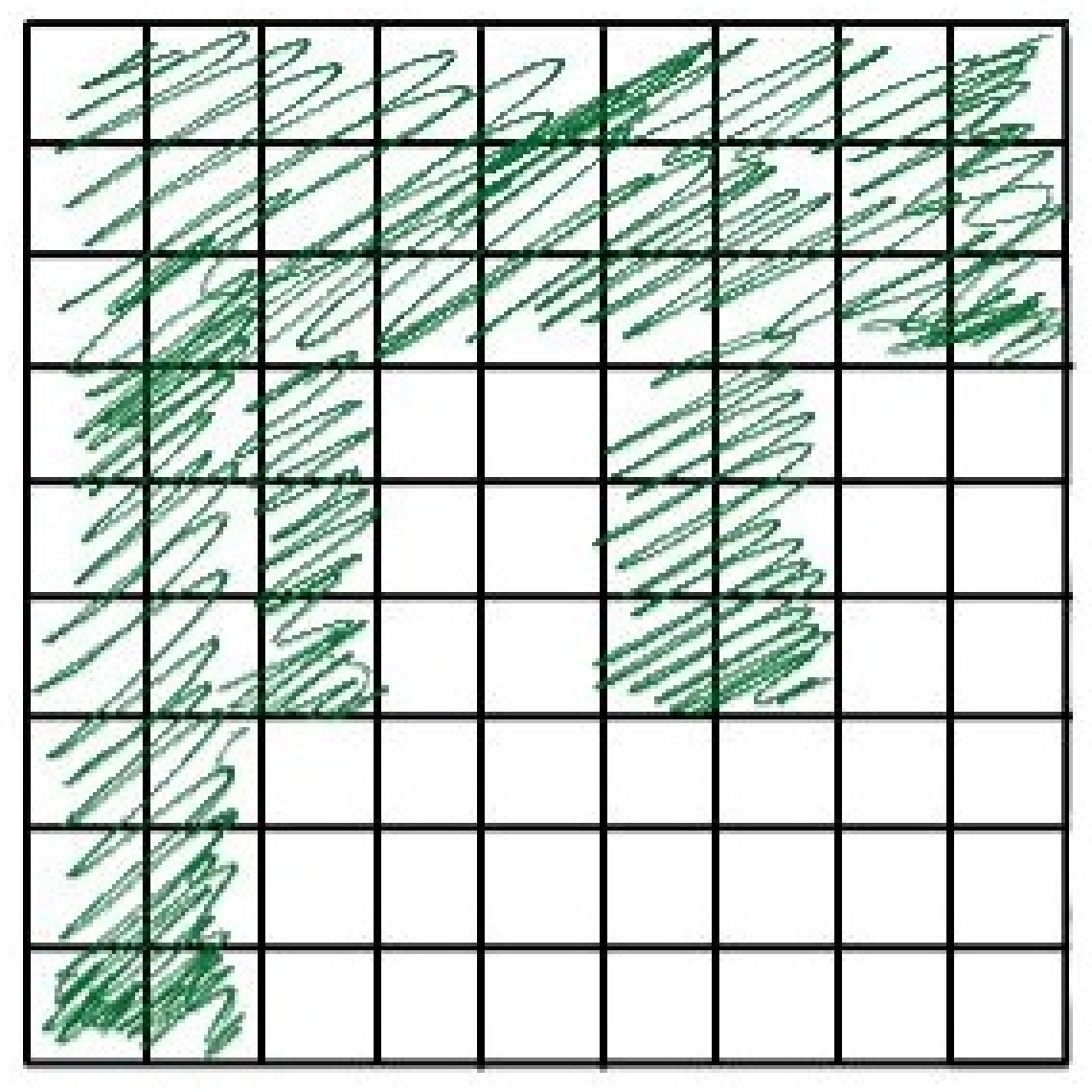}
			\caption[]{A point in $K_2$.  All green squares have one value and all white squares have another.}
						\label{fig:blob}
		\end{subfigure}
		\quad
		\begin{subfigure}[b]{.25 \textwidth}

			\includegraphics[width = \textwidth]{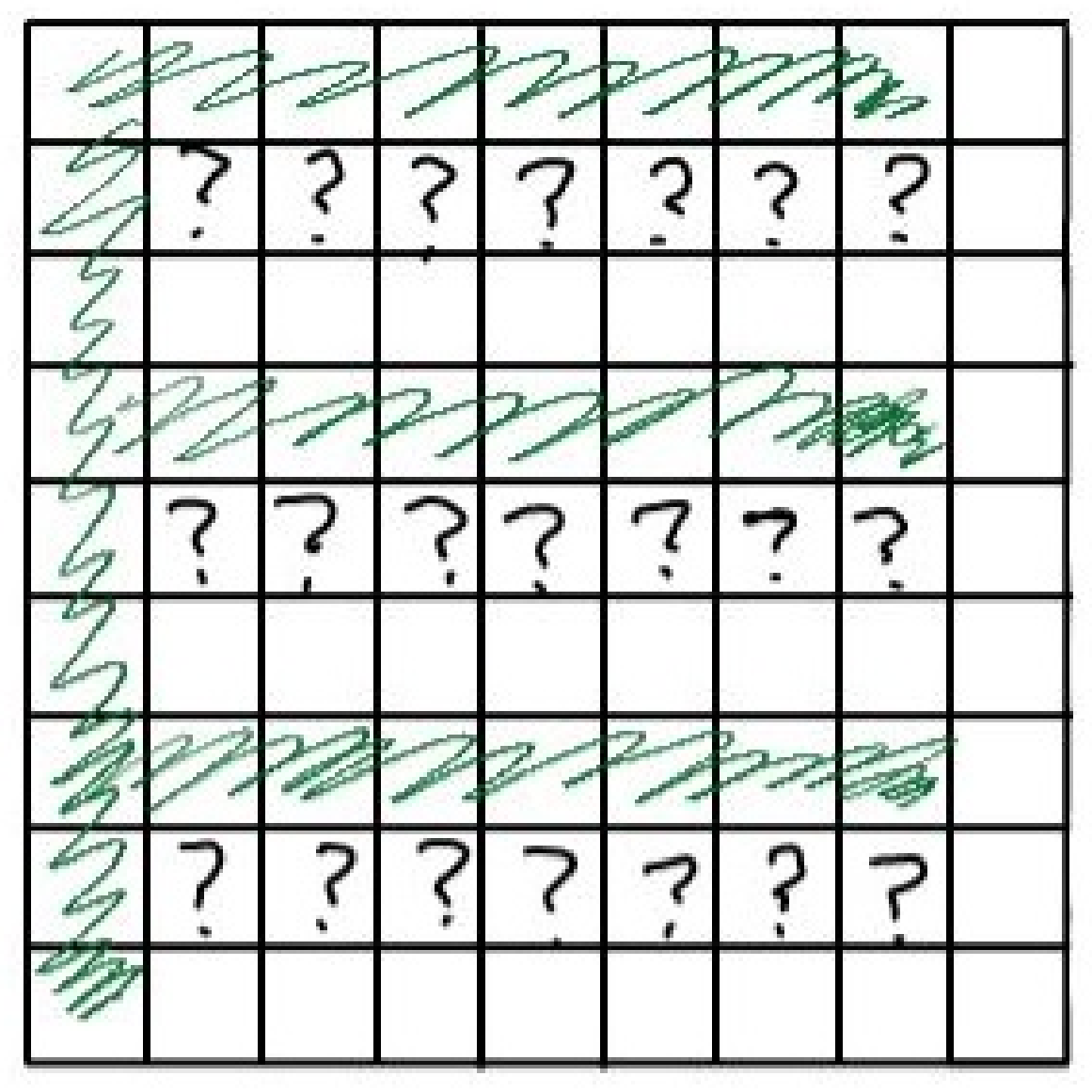}
			\caption[]{Visualization of packing patterns.  Green corresponds with $1/n$, white corresponds with $-1/n$, and ? can be either.}
						\label{fig:packingpic}
		\end{subfigure}
		\quad
		\begin{subfigure}[b]{.25 \textwidth}
			\includegraphics[width = \textwidth]{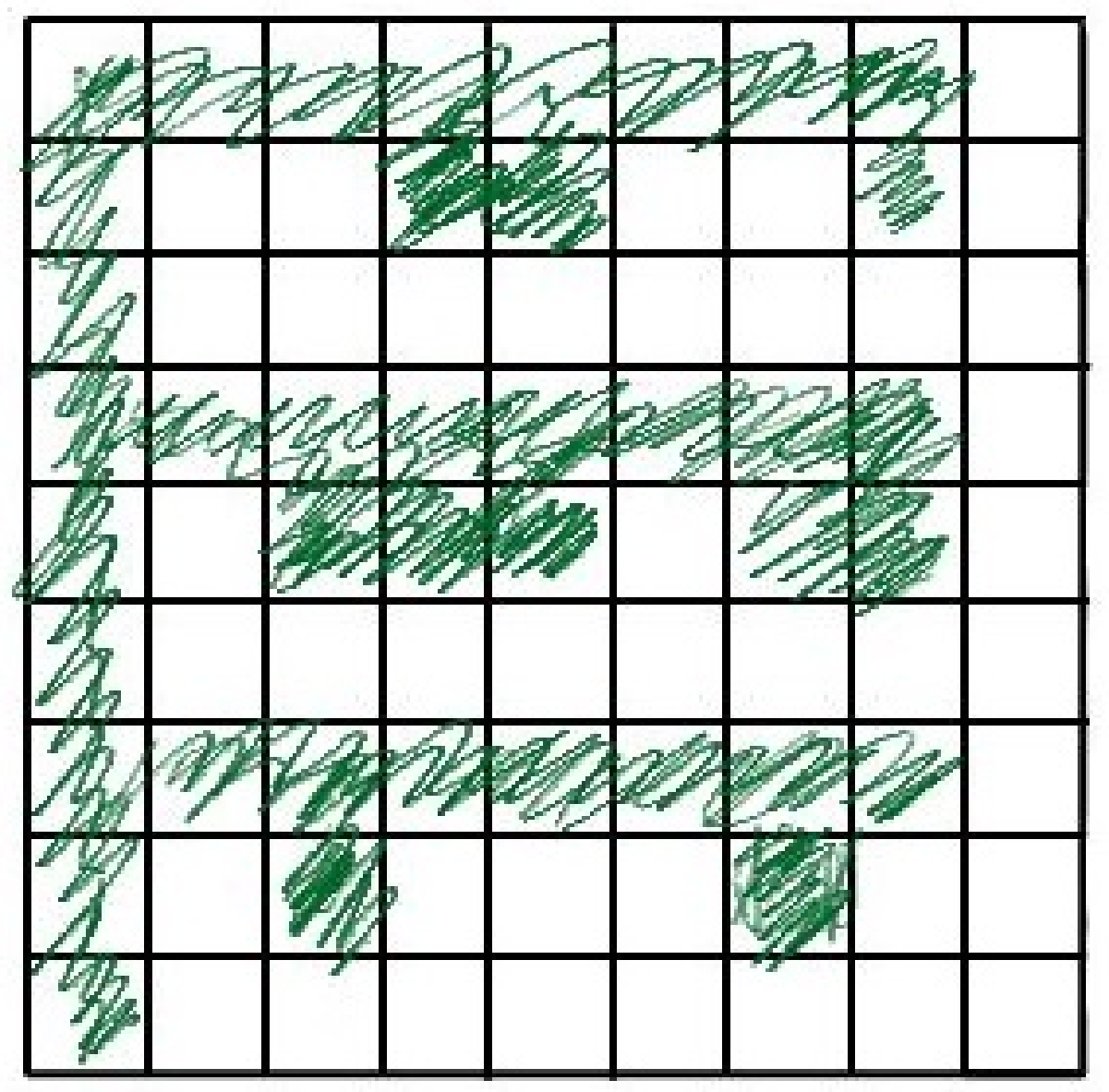}
			\caption[]{One possible point in the packing.}
				\label{fig:packingpic_oneconfig}
		\end{subfigure}
		\caption[]{}
\end{figure}

Note that each entry of our blobs has magnitude $1/n$.  Each of the $q:= n(n-2)/3$ question marks can take the value $1/n$ or $-1/n$.  Thus, our problem reduces to constructing a packing of the Hamming cube
\[\tilde{K}:= \frac{1}{n} \{+1, -1\}^q.\]
Such packings are well known.  We use a random construction.

We pack $\tilde{K}$ by randomly picking a number of points in it and showing that with nonzero probability each pairwise distance is at least 1/2.  Thus, let $\x$ and $\x'$ be two points picked uniformly at random from $\tilde{K}$.  Note that $\twonorm{n \cdot \x - n \cdot \x'}^2 \sim 4\cdot \text{Binomial}(q, 1/2)$.  It follows from Hoeffding's inequality that
\begin{eqnarray}
 \Pr{\twonorm{ \x -  \x'}^2 < \frac{q}{d}} = \Pr{\twonorm{n \cdot \x - n \cdot \x'}^2 < q} \leq  \exp(-q/8).
\end{eqnarray}
We conclude the construction of the packing by applying Lemma \ref{lem:random packing}.  In using the lemma, note that we assumed $d \geq 64$, and thus $q = n(n-2)/3\geq d/4$.
\end{proof}

\subsection{Gaussian $\OM$}

We now construct a packing when $\OM$ is Gaussian.

In this section we take $\OM \in \R^{p \times d}$ to have $\textit{i.i.d.}$ standard normal entries.
Note that with probability 1 the rows of $\OM$ are in general position and thus
\[K_b := \bigcup_{\abs{T} = p - (d - b)} K_T.\]

Let $K := K_1 \cap S^{d-1}$.  We construct a packing for $K$.

\begin{lemma}
\label{lem: pack Gaussian}
\[P(K, 1/2) \geq 
 \frac{1}{\sqrt{3}}  \exp\left(\frac{d-1}{8} \cdot \left(1 - \frac{d-2}{p} \right)\right).\]
\end{lemma}
\begin{proof}
As above, we randomly construct our packing.  We will pick a series of points in $K$ as follows.  First, pick $d-1$ rows of $\OM$ uniformly at random indexed by the set $\Lambda \in [p]$.  Take $\x$ which satisfies $\OM_{\Lambda} = 0$.  This restricts $\x$ to a 1-dimensional space.  Since $\x$ must be on the unit sphere, there are only two possible points.  Pick one at random.

Take $\x$ and $\x'$ to be two randomly generated points; we must show that they are far apart with high probability. Let $\Lambda$ and $\Lambda'$ be the corresponding rows of $\OM$, so that $\OM_{\Lambda} \x = 0$ and $\OM_{\Lambda'} \x' = 0$.  Note that $\x$ and $\x'$ are both drawn uniformly from the sphere, but they are not independent because $\Lambda$ and $\Lambda'$ may have some intersection.  The proof will follow by controlling their dependence.  First, we note that $\Lambda$ and $\Lambda'$ do not have an overly large intersection (with high probability).

\begin{lemma}[Bounding $\abs{\Lambda \cap \Lambda'}$]
\label{lem:overlap}
\begin{eqnarray}
 \Pr{\abs{\Lambda \cap \Lambda'} \geq (d-1) \cdot \left(\frac{d + p}{2p}\right)} 
 \leq \exp\left(-\frac{(d-1)(p - d+2)}{2p}\right)
\end{eqnarray}
\end{lemma}
\begin{proof}
 An application of Corollary 1.1 in \cite{Serfling74Probability}  (an extension of Hoeffding's inequality for sampling without replacement \cite{Hoeffding63Probability}) gives
 \begin{eqnarray}
\Pr{\abs{\Lambda \cap \Lambda'}  \geq \frac{(d-1)^2}{p} + t} 
\leq \exp\left(-\frac{2t^2}{(d-1)(1-\frac{d-2}{p})}\right).
\end{eqnarray}
The lemma follows by taking $t = (d-1)(p - d+2)/(2 p)$.
\end{proof}

Now condition on $\Lambda, \Lambda'$, and $\OM_{\Lambda \cap \Lambda'}$.  Both $\x$ and $\x'$ must satisfy $\OM_{\Lambda \cap \Lambda'} \x = 0$, thus reducing the dimension of the space they live in to $d - \abs{\Lambda \cap \Lambda'}$.  Set $T := \Lambda \setminus \Lambda'$ and $T' := \Lambda' \setminus \Lambda$.  Note that $\OM_T$ and $\OM_{T'}$ are independent.  Thus, by the rotational invariance of the Gaussian distribution, $\x$ and $\x'$ are distributed uniformly at random in the orthogonal complement to span($\OM_{\Lambda \cap \Lambda'}$).  The distance between $\x$ and $\x'$ is equal in distribution to the distance between two points chosen uniformly at random on $S^{q-1}$ where $q := d - \abs{\Lambda \cap \Lambda'}$.  Let $z := \twonorm{\x - \x'}$.  Note that the geodesic distance between $\x$ and $\x'$ is equal to $2 \arcsin(z)$.  The distribution of $z$ does not change if we fix $\x'$, and thus, the probability that $z \leq 1/2$ is precisely the normalized measure of a spherical cap with geodesic radius $2 \arcsin(1/2) \leq 0.52$.  Bounds for this quantity are well known (see \cite{Talagrand91Probability}[Theorem 1.1]) giving
\begin{eqnarray}
\Pr{ \twonorm{\x - \x'} \leq \frac{1}{2}} \leq 2 \exp\left(-\frac{q}{2}\right).
\end{eqnarray}

We only need to bound $q$, but this is done in Lemma \ref{lem:overlap}.  Let $E$ be the good event that $\abs{\Lambda \cap \Lambda'} \geq (d-1) \cdot \left(\frac{d + p}{2p}\right)$.  By Lemma \ref{lem:overlap}, $\Pr{E} \geq 1 - \exp\left(-\frac{(d-1)(p - d+2)}{2p}\right)$.  On the event $E$, we have $q \geq d - (d-1) \cdot \left(\frac{d + p}{2p}\right) = 1+ (d-1) \cdot \left(\frac{p - d}{2p}\right)$.

Putting these pieces together, we get
\begin{eqnarray*}
\Pr{\twonorm{\x - \x'} \leq \frac{1}{2}} 
&\leq & \Pr{\twonorm{\x - \x'} \leq \frac{1}{2} \mid E} + \Pr{E^c}\\
& \leq & 2 \exp\left(-\frac{1}{2} - \frac{(d-1)(p-d)}{4p}\right)  + \exp\left(-\frac{(d-1)(p - d+2)}{2p}\right)\\  &\leq& 3  \exp\left(-\frac{(d-1)(p - d+2)}{4p}\right),
\end{eqnarray*}
where in the last line we have used the fact that $\frac{2(d-1)}{4p} \le \frac{1}{2}$.

Complete the proof by applying Lemma \ref{lem:random packing} with $\delta = 1/2$ and $\eta$ as in the last line of the above equation.
\end{proof}

\subsection{Implications of Set Packing}

Consider a vector $\x$ which is known to reside in a set $K$.  We show that if $K$ admits a large packing, then $\x$ cannot be robustly reconstructed from few linear measurements by any method.  The proof proceeds by showing that the distance between some pair of points in the packing will be reduced immensely when subsampling, and thus the corresponding two points in $K$ will be nearly indistinguishable amid noise.

We need the following lemma in the proof.  This lemma is a classical result about packing numbers. In this lemma and the ones to follow, we denote by $B^n$ \rg{($\subset \R^{n}$)} the $\ell_2$ ball of radius 1 centered at the origin.

\begin{lemma}[Minimum distance in a packing]
\label{lem:packing_dist}
Let $\X \subset B^m$ be a finite set of points.  Then
\[\min_{\stackrel{x \neq y}{x,y \in \X}} \twonorm{\x - \y} \leq \frac{4}{\abs{\X}^{1/m}}.\]

\end{lemma}

The proof is a simple and classical volumetric argument, see \cite{pisier1999volume}.     

The following lemma begins to address the problem of signal estimation.
\begin{lemma}
\label{lem:fails on packing}
Let $\X \subset B^d$ be a finite set of points.  Suppose $\y = \A \x + \z$ for some $\x \in \X$, $\A \in \R^{m \times d}$ with $\opnorm{\A} \leq 1$, and $\z \sim \NN(0, \sigma^2\cdot \Id)$.  Assume
\[\frac{4}{\abs{\X}^{1/m}} \leq \sigma.\]
Let $\hat{\x} = \hat{\x}(\y)$ be any estimator of $\x$.  Then,
\begin{equation}
\label{eq:minimax_prob}
\min_{\x \in \X} \Pr{\hat{\x} = \x} \leq \frac{3}{4}.
\end{equation}
\end{lemma}
\begin{proof}
We will lower-bound the worst-case probability of error by the probability of error under a suitably unfavorable prior. (This reduction from minimax to Bayesian is a standard trick, see, for example, \cite[Equation 2]{guntuboyina}). 

First, apply the packing bound (Lemma \ref{lem:packing_dist}) to $\A \X$ to show that there are some two points $\A \x_1, \A \x_2 \in \A \X$ satisfying
\[\twonorm{\A \x_1 - \A \x_2} := \eps \leq \frac{4}{\abs{\X}^{1/m}} \leq \sigma. \]
Consider the prior distribution which picks $\x_1$ with probability $1/2$ and $\x_2$ with probability $1/2$.  For any given estimator, the worst-case probability of error (the left-hand side of Equation \eqref{eq:minimax_prob}) is lower-bounded by the probability of error under this prior.  This is further minimized by the \textit{Bayes Estimator} which chooses $\x_1$ or $\x_2$ based on which has the highest posterior probability conditional on $\y$.  The Bayes estimator simply takes
\[\hat{\x} = \arg\min_{\x_1, \x_2} \twonorm{\A \x - \y}.\]
It is straightforward to show that this estimator \rg{satisfies},
\begin{eqnarray}
\Pr{\hat{\x} = \x}  =  \Pr{\NN(0,1) \leq   \frac{\eps}{2 \sigma}}  
 \leq  \Pr{\NN(0,1) \leq \frac{1}{2}} \leq \frac{3}{4}.
\end{eqnarray}
\end{proof}

The following proposition is the synthesized tool relating packings to minimax error that we will use to prove our main theorems.
\begin{proposition}
\label{prop:main proposition}
Let $K \subset \R^d$ be a cone.  Let $\X$ be a $\delta$-packing of $K \cap B^d$.  Suppose $\y = \A \x + \z$ for $\x \in K$, $\A \in \R^{m \times d}$ with $\opnorm{\A} \leq 1$, and $\z \sim \NN(0, \sigma^2\cdot \Id)$.  Then for any estimator $\hat{\x} = \hat{\x}(\y)$, we have
\[\sup_{\x \in K} \E \twonorm{\hat{\x} - \x} \geq \frac{\delta \sigma \abs{\X}^{1/m}}{32} .\]
\end{proposition}
\begin{proof}
We begin by rescaling the problem so that the noise level is just large enough that a signal in $\X$ will be hard to recover, i.e., so that we may use Lemma \ref{lem:fails on packing}.  Let
\[\lambda := \frac{4}{\abs{\X}^{1/m} \sigma}\]
and set $\tilde{\y} = \lambda \y, \tilde{\x} = \lambda \x,$ and $\tilde{\z} = \lambda \z$.  Thus, $\tilde{\y} = \A \tilde{\x} + \tilde{\z}$ and $\tilde{\z}, \X$ satisfy the conditions of Lemma \ref{lem:fails on packing}.  Note also that $\tilde{\x} \in \lambda K = K$.  We further restrict $\tilde{\x}$ to lie in $\X$.

Now, by Lemma \ref{lem:fails on packing}, for any estimator $\hat{\x} = \hat{\x}(\y)$,
\[\min_{\tilde{\x} \in \X}\Pr{\hat{\x}(\y) = \tilde{\x}} \leq 3/4.\]
Since no estimator can reliably guess $\tilde{\x}$ on a $\delta$-packing, it follows that no estimator can estimate $\tilde{\x}$ to accuracy better than $\delta/2$ with high probability.  Otherwise, such an estimator could be projected onto $\X$ to make a reliable guess.  In other words, 
\[\min_{\tilde{\x} \in \X} \Pr{\twonorm{\hat{\x}(\y) - \tilde{\x}} < \delta/2} \leq 3/4.\]
This implies that 
\[\max_{\tilde{\x} \in \X} \E \twonorm{\hat{\x} - \tilde{\x}} \geq \frac{\delta}{8} .\]
Divide both sides of the equation by $\lambda$ to undo the scaling and finish the proof.
\end{proof}

\begin{figure}[t]
  \centering
 \centerline{\includegraphics[width=10cm]{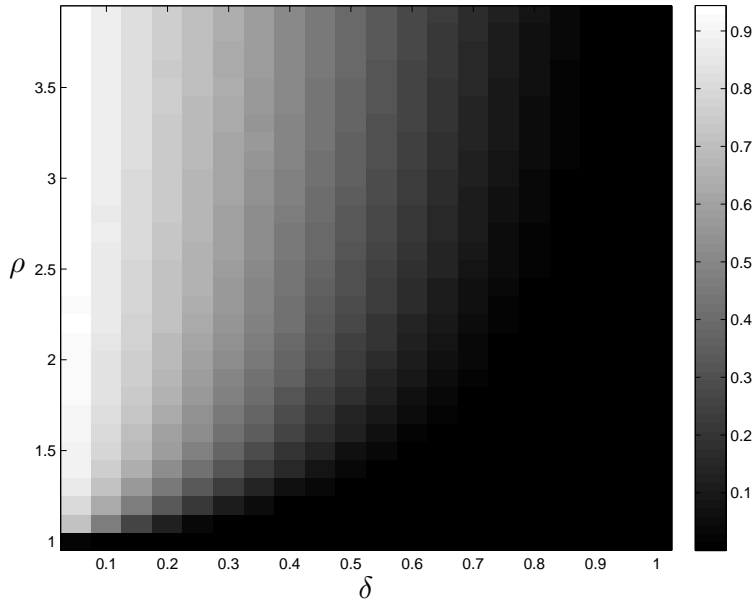}}
\caption{Recovery error of $\ell_1$-minimization with Gaussian analysis dictionary in the noiseless case. Experiment setup: $d=200$, $\delta = \frac{m}{d}$ and $\rho = \frac{p}{d}$. For each configuration we average over $500$ realizations.
Color attribute: Mean Squared Error.
  }
\label{fig:l1_gauss_recovery}
\end{figure}

\begin{remark}[Interpretation in terms of complexity, or metric dimension, of $K$]
While the manifold dimension of a signal set determines signal recoverability in the noiseless case \cite{Eldar12Uniqueness}, it can fail to characterize the noisy case.  Instead, a classical metric notion of dimension, following ideas of Kolmogorov and Le Cam, provides a more apt characterization.  Indeed, set $D(K) := \log(P(K\cap B^d, 1/2))$.  Le Cam \cite{le1986asymptotic} showed that $D(K)$ is a effective metric characterisation of the complexity (or dimension) of the set $K$ in regards to many point estimation problems.  As a simple example, if $K$ is a $q$-dimensional subspace it is well-known that $D(K)$ is proportional to $q$, just like the manifold dimension.  However, in contrast to the manifold dimension, this metric definition takes into account the geometry of the set, thus allowing characterization of signal recoverability amid noise.  The above proposition states that signal recovery error is at least proportional to $\exp(D(K)/m)$, i.e., if the number of measurements is below the effective dimension, the error amid noise blows up exponentially fast as a function of the ratio. 
\end{remark}
\subsection{Putting it together}
\label{sec:finals steps}
Theorem \ref{thm:TV_main_theorem} now follows by combining the packing number of Lemma \ref{lem:pack finite differences} with Proposition \ref{prop:main proposition}.  Theorem \ref{thm:Gaussian_main_theorem} follows by combing Lemma \ref{lem: pack Gaussian} with Proposition \ref{prop:main proposition}.

\begin{figure}[t]
  \centering
 \centerline{\includegraphics[width=10cm]{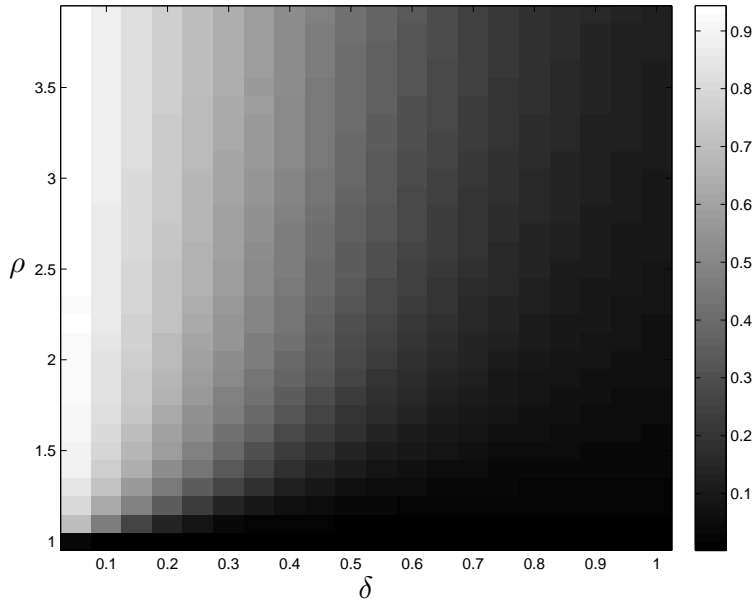}}
\caption{Recovery error of $\ell_1$-minimization with Gaussian analysis dictionary in the noisy case with $\sigma = 0.01$. Experiment setup: $d=200$, $\delta = \frac{m}{d}$ and $\rho = \frac{p}{d}$. For each configuration we average over $500$ realizations.
Color attribute: Mean Squared Error.}
\label{fig:l1_gauss_denoising}
\end{figure}

\section{Experiments}
\label{sec:exp}

To demonstrate the results of the theorems, in this section we look at the performance of analysis $\ell_1$-minimization 
\begin{eqnarray}
\label{eq:analysis_l1_noiseless}
\min_{\x'}\onenorm{\OM\x'} & s.t.& \twonorm{\y - \A\x'} \le \sqrt{m}\sigma,
\end{eqnarray}
in recovering signals with low dimensionality and different cosparsity levels. In all the experiments the measurement matrix $\A$ is a random Gaussian matrix with normalized columns. 

In the first experiment, we select $\OM$ to be a random Gaussian matrix and the signal $\x$ to be a Gaussian random vector projected to a one dimensional subspace orthogonal to randomly selected $d-1$ rows from $\OM$.  
In Fig.~\ref{fig:l1_gauss_recovery} we present the recovery performance in the noiseless case for a fixed signal ambient dimension $d=200$ and 
 different combinations of the sampling rate $\delta = \frac{m}{d}$ and the redundancy ratio $\rho = \frac{p}{d}$.  Interestingly, \rg{we observe empirically} that the theoretical instability to noise (Theorem \ref{thm:Gaussian_main_theorem}) also implies instability to $\ell_1$ relaxation. 
Indeed, notice that though the manifold dimension of the signal is equal to 1 in all the experiments, the success in recovery heavily depends on $p$, which changes only the cosparsity of the signal. As expected from the theory, as soon as $p$ increases to be slightly larger than $d$, the number of measurements needed to reconstruct the signal increases enormously.

In Fig~\ref{fig:l1_gauss_denoising} we present the reconstruction error in the noisy case when an additive random white Gaussian noise with standard deviation $\sigma = 0.01$ is added to the measurements. 
Notice that the error is saturated by $1$, the signal energy, as when the noise is very large the best estimator is the zero estimator, for which the error equals to the signal energy.

\begin{figure}[t]
  \centering
  \centerline{\includegraphics[width=10cm]{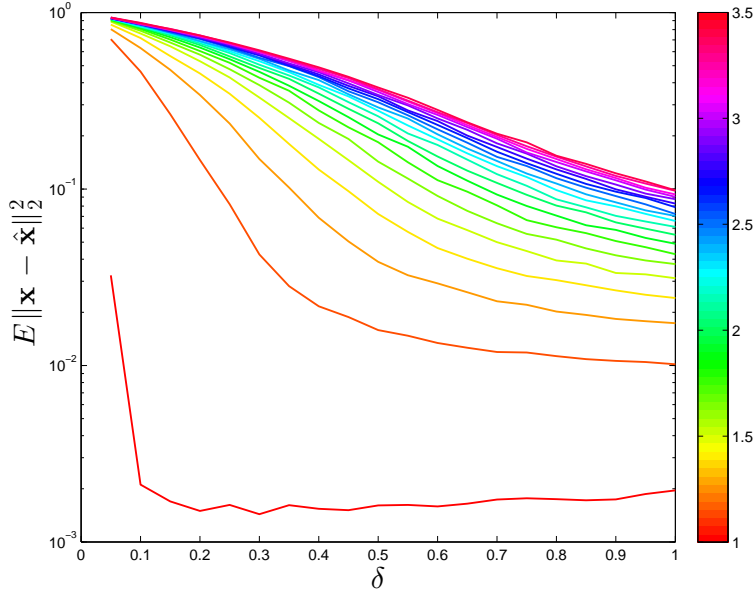}}
\caption{Recovery error of $\ell_1$-minimization with Gaussian analysis dictionary in the noisy case with $\sigma = 0.01$ as a function of $\delta = \frac{m}{d}$ for different selections of $\rho = \frac{p}{d}$. The signal dimension is $d=200$ and we average over $500$ realizations.
Color attribute: The color of each graph corresponds to the value of $\rho$. The bottom graph corresponds to $\rho = 1$ and the upper one to $\rho = 3.5$.}
\label{fig:l1_gauss_denoising_vs_m}
\end{figure}

As predicted from the theorem, when $\rho > 1$
the recovery error grows exponentially as $m$, the number of measurements, decreases. This becomes clearer for larger values of $\rho$.  
To show it more clearly we present in Fig.~\ref{fig:l1_gauss_denoising_vs_m} the recovery error as a function of the sampling rate, which is a function of $m$. The bottom graph corresponds to $\rho = 1$, where the manifold dimension equals $d$ minus the cosparsity, for which we get a good recovery almost for all values of $m$. This is not surprising because when $d = m$, $\Omega$ is invertible and the analysis cosparse model may be recast into the standard synthesis model, in which one expects $O(log (p))$ measurements to suffice.  Indeed, we see that already when we have ten percent of the measurements (corresponds to $20 > 2\log(p) \simeq 10.6$ measurements) we get a very good recovery. However, the behavior is quite different as soon as $\rho$ increases to slightly greater than 1.  The rest of the graphs are above it and ordered according to increasing values of $\rho$. As $\rho$ becomes larger the error increases and its behavior as a function of $m$ becomes more and more exponential. 

\begin{figure}[t]
  \centering
 \centerline{\includegraphics[width=10cm]{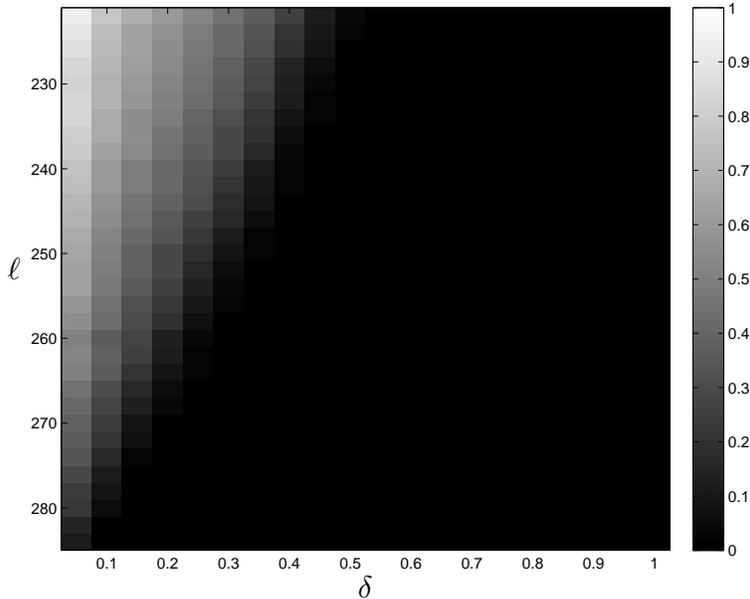}}
\caption{Recovery error of $\ell_1$-minimization with the 2D-DIF analysis dictionary in the noiseless case. Experiment setup: $d=144$, $\delta = \frac{m}{d}$ and $\cosp$ is the cosparsity level. For each configuration we average over $500$ realizations.
Color attribute: Mean Squared Error.
  }
\label{fig:l1_TV_recovery}
\end{figure}

In the second experiment we consider the 2D-DIF operator.
Notice that for this operator we may have  different cosparsity levels for the same manifold dimension. 
We consider vectors of size $d =144$ that represent two dimensional images of size $12 \times 12$ with manifold dimension two.
An image of dimension two is an image with two connected components, each with a different gray value. We generate randomly such images with different cosparsity levels. The values in the first and second connected components are selected randomly from the ranges $[0,1]$ and $[-1,0]$ respectively. Note that the cosparsity level defines the length of the edge in the image. The images are generated by setting all the pixels in the image to a value from the range $[0,1]$ and then picking one pixel at random and starting a random walk (from its location) that assigns to all the pixels in its path a value from the range $[-1, 0]$ (the same value). The random walk stops once it gets to a pixel that it has visited before. With high probability the resulted image will be with only two connected components. We generate many images like that and sort them according to their cosparsity (eliminating images that have more than two connected components). Note that the larger the cosparsity the shorter the edge.

We test the reconstruction performance for different cosparsity levels and sampling rates. 
The recovery error in the noiseless and noisy cases are presented in Figs.~\ref{fig:l1_TV_recovery} and \ref{fig:l1_TV_denoising}.
It can be clearly seen that the recovery performance in both cases is determined by the cosparsity level and not the manifold dimension of the signal which is fixed in all the experiments. 
As predicted by Theorem~\ref{thm:TV_main_theorem}, if we rely only on the manifold dimension we get a very unstable recovery. However, if we take into account also the cosparsity level we can have a better prediction of our success rate.
As we have seen in the previous experiment, also here we can see that instability in the noisy case also implies instability to relaxation.

\begin{figure}[t]
  \centering
 \centerline{\includegraphics[width=10cm]{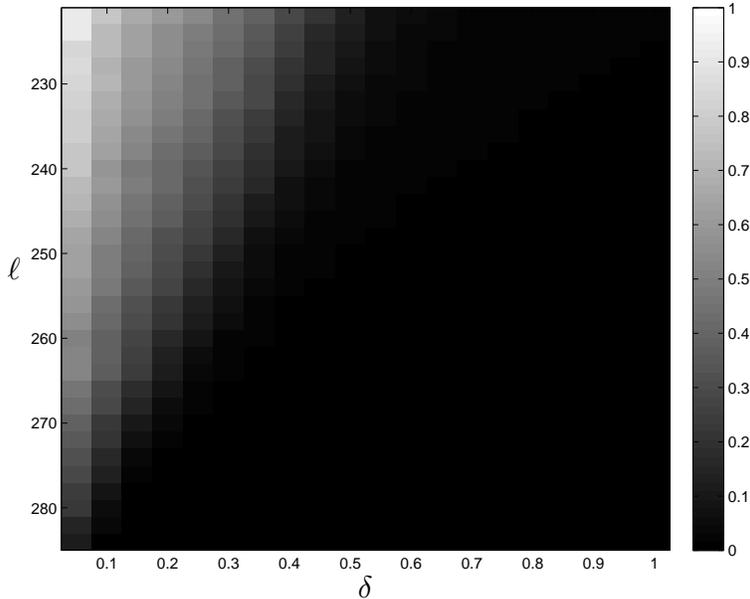}}
\caption{Recovery error of $\ell_1$-minimization with the 2D-DIF dictionary in the noisy case with $\sigma = 0.01$. Experiment setup: $d=144$, $\delta = \frac{m}{d}$ and $\cosp$ is the cosparsity level. For each configuration we average over $500$ realizations.
Color attribute: Mean Squared Error.}
\label{fig:l1_TV_denoising}
\end{figure}

\section{Conclusion}
\label{sec:conc}

In this work we have inquired whether it is possible to provide recovery guarantees for compressed sensing with signals from the cosparse analysis framework by only having information about their manifold dimension. Though the answer for this question is positive for standard compressed sensing (with the standard sparsity model) and the matrix completion problem, we have shown that this is not the case here. We have demonstrated this with two analysis dictionaries, the Gaussian matrix and the 2D-DIF operator, both in theory and simulations. Our conclusion is that in the cosparse analysis framework the ``correct" measure to use for predicting the recovery success of any tractable method is the cosparsity of the signal (number of zeros) and not the manifold dimension.

\rg{It would be interesting to check whether it is possible to carry over the results in this paper to non-subspace models such as curved manifolds. Notice that some theoretical guarantees have  already been given for the latter case \cite{baraniuk2009random, wakin2010manifold}, showing that it is possible to recover a $b$-dimensional submanifold from a number of samples proportional to $b$. However, these results also depend on a number of other quantities, such as the volume and condition number of the manifold. Therefore, another direction would be to see whether adding similar assumptions to the analysis model would allow sampling in the manifold dimension. }

\rg{An additional open question raised by this work is whether instability in the noisy case leads to instability to relaxation. Indeed, we have observed this phenomenon empirically in the experiments and therefore believe  
that there is a room to prove such a result.}


\section*{Acknowledgment}
The authors thank the reviewers of the manuscript for their suggestions which improved the paper.

\bibliographystyle{IEEEtran}
\bibliography{negative, IEEEabrv}

\end{document}